%% file: main.tex
\documentclass[a4paper,11pt,noamsfonts]{amsart}
\pdfoutput=1
\usepackage[utf8]{inputenc}
\usepackage[english]{babel}
\usepackage[T1]{fontenc}

\usepackage{newtx}

\usepackage{microtype}
\usepackage{color}
\usepackage{pgfplots}
\pgfplotsset{compat=1.18}
\usepackage{tikz}
\usetikzlibrary{quantikz2}
\usetikzlibrary{cd,pgfplots.statistics}

\usepackage{standalone}
\usepackage[numbers]{natbib}
\usepackage[hidelinks]{hyperref}
\usepackage{bbold}
\usepackage{bm}
\usepackage{caption}
\usepackage{subcaption}
\usepackage[a4paper, total={150mm,215mm}, centering]{geometry}

\newtheorem{theorem}{Theorem}

\newtheorem{lemma}[theorem]{Lemma}
\theoremstyle{definition}
\newtheorem{definition}[theorem]{Definition}
\theoremstyle{remark}
\newtheorem{example}[theorem]{Example}
\newtheorem{remark}[theorem]{Remark}

\definecolor{mycolor1}{rgb}{0.00000,0.44700,0.74100}
\definecolor{mycolor2}{rgb}{0.85000,0.32500,0.09800}
\definecolor{mycolor4}{rgb}{0.92900,0.69400,0.12500}

\DeclareFontFamily{U}{matha}{\hyphenchar\font45}
\DeclareFontShape{U}{matha}{m}{n}{
	<-6> matha5 <6-7> matha6 <7-8> matha7
	<8-9> matha8 <9-10> matha9
	<10-12> matha10 <12-> matha12
}{}
\DeclareSymbolFont{matha}{U}{matha}{m}{n}

\DeclareFontFamily{U}{mathx}{\hyphenchar\font45}
\DeclareFontShape{U}{mathx}{m}{n}{
	<-6> mathx5 <6-7> mathx6 <7-8> mathx7
	<8-9> mathx8 <9-10> mathx9
	<10-12> mathx10 <12-> mathx12
}{}
\DeclareSymbolFont{mathx}{U}{mathx}{m}{n}

\DeclareMathDelimiter{\vvvert} {0}{matha}{"7E}{mathx}{"17}

\newcommand{\N}{\mathbb{N}}
\newcommand{\Z}{\mathbb{Z}}
\newcommand{\R}{\mathbb{R}}
\newcommand{\C}{\mathbb{C}}
\newcommand{\E}{\mathbb{E}}

\let\rv\bm
\newcommand{\calO}{\mathcal{O}}

\newcommand{\Nseed}{M}
\newcommand{\tgrf}{\rv{Z}}
\newcommand{\wn}{\rv{W}}
\newcommand{\calA}{\mathcal{A}}
\newcommand{\calC}{\mathcal{C}}
\newcommand{\RNG}{\mathcal{R}}
\newcommand{\QFT}{\mathcal{F}}

\newcommand{\dW}{\, \mathrm{d}\rv{W}(s)}
\newcommand{\rh}{{r\!,h}}
\newcommand{\anc}{a}
\newcommand{\distN}{\mathcal{N}}
\newcommand{\distB}{\mathcal{B}}

\DeclareMathOperator{\polylog}{polylog}
\DeclareMathOperator{\poly}{poly}

\DeclareMathOperator{\Id}{Id}
\DeclareMathOperator{\Cov}{Cov}
\DeclareMathOperator{\Var}{Var}
\DeclareMathOperator{\Diag}{Diag}
\DeclareMathOperator{\sinc}{sinc}

\newcommand{\dx}{\, \mathrm{d}x}
\newcommand{\tol}{\varepsilon}

\renewcommand{\ket}[1]{|#1\rangle}
\renewcommand{\Pr}{\mathbb{P}}
\renewcommand{\Re}{\operatorname{Re}}
\renewcommand{\Im}{\operatorname{Im}}

\allowdisplaybreaks
\makeatletter
\def\paragraph{\@startsection{paragraph}{4}%
  {\z@}{2ex}{-1em}%
  {\normalfont\bfseries}}
\makeatother

\makeatletter
\renewcommand{\qwbundle}[2][]{%
  \pgfkeys{/quantikz/gates/.cd,style=,Strike Width=0.08cm,Strike Height=0.12cm,#1}%
  \pgfkeysgetvalue{/quantikz/gates/style}{\qz@style}%
  \pgfkeysgetvalue{/quantikz/gates/Strike Width}{\qz@sw}%
  \pgfkeysgetvalue{/quantikz/gates/Strike Height}{\qz@sh}%
  \expanded{%
    \noexpand\arrow[strike arrow={\qz@sw}{\qz@sh}{\unexpanded{#2}},\qz@style,phantom]{l}%
  }%
}
\makeatother

\begin{document}

\title[Quantum Sampling and Moment Estimation]{Quantum Sampling and Moment Estimation for Transformed Gaussian Random Fields}
\author[Matthias Deiml, Daniel Peterseim]{Matthias Deiml$^1$, Daniel Peterseim$^{1,2}$}
\address{$^1$Institute of Mathematics, University of Augsburg, Universit\"atsstr.~12a, 86159 Augsburg, Germany}
\address{$^2$Centre for Advanced Analytics and Predictive Sciences (CAAPS), University of Augsburg, Universit\"atsstr.~12a, 86159 Augsburg, Germany}
\email{\{matthias.deiml,daniel.peterseim\}@uni-a.de}
\thanks{The authors acknowledge support by the Deutsche Forschungsgemeinschaft (DFG, German Research Foundation) through the project~496984632.}
\date{\today}
\begin{abstract}
We present a quantum algorithm for efficiently sampling transformed Gaussian random fields on $d$-dimensional domains, based on an enhanced version of the classical moving average method. Pointwise transformations enforcing boundedness are essential for using Gaussian random fields in quantum computation and arise naturally, for example, in modeling coefficient fields representing microstructures in partial differential equations. Generating this microstructure from its few statistical parameters directly on the quantum device bypasses the input bottleneck. Our method enables an efficient quantum representation of the resulting random field and prepares a quantum state approximating it to accuracy $\tol > 0$ in time $\mathcal{O}(\polylog \tol^{-1})$. Combined with amplitude estimation and a quantum pseudorandom number generator, this leads to algorithms for estimating linear and nonlinear observables, including mixed and higher-order moments, with total complexity $\mathcal{O}(\tol^{-1} \polylog \tol^{-1})$. We illustrate the theoretical findings through numerical experiments on simulated quantum hardware.
\end{abstract}

\maketitle

{
}

\input{content.tex}

\bibliographystyle{quantum}
\bibliography{references}

\appendix

\input{appendix}

\end{document}

%% file: content.tex
\section{Introduction}
Recent work has shown that quantum algorithms can outperform classical ones for certain high-dimensional tasks, including solving large linear and nonlinear algebraic systems~\cite{GSLW19,MPS+25,RR23,DP24a} and, in particular, suitably discretized partial differential equations (PDEs)~\cite{CPP+13,Ber14,BCOW17,CL20,CLO21,BNWA23a,HJZ24,LOC24,GHL25,DP25}. However, all such results run into a quantum information bottleneck: to remain computationally efficient, the amount of input data that must be loaded into the quantum device has to be very small relative to the overall problem size.
This limitation underscores the need to develop quantum-native methods that generate or preprocess input data directly on the quantum computer. Random coefficient fields in a PDE are a prime example: while individual realizations are too large to load explicitly, the law of the field is often described by only a few parameters within a parametric family. If sampling from that law can be performed on the quantum hardware, the PDE solvers built on top of it have a genuine chance of achieving quantum efficiency. 
\begin{figure}
    \centering
    \includegraphics[width=0.25\linewidth]{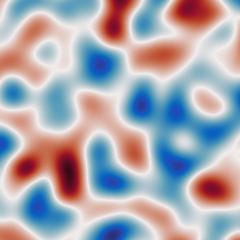}
    \hspace{0.1\linewidth}
    \includegraphics[width=0.25\linewidth]{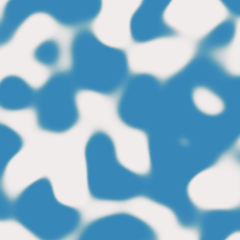}
    \hspace{0.1\linewidth}
    \begin{tikzpicture}
    \useasboundingbox (0,0) rectangle (0.0415\linewidth,0.25\linewidth);
    \node[inner sep=0pt, anchor=south west] at (0,0) {\includegraphics[height=0.25\linewidth]{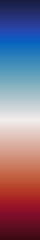}};
    \begin{axis}[
    at={(0,0)},
    anchor=south west,
	scale only axis,
    width=0.0415\linewidth,
    height=0.25\linewidth,
    xmin=0,
    xmax=1,
    ymin=-22,
    ymax=22,
    xtick=\empty,
    ytick={-22,0,22},
    yticklabel style={font=\small},
    axis lines=right,
    axis line style={draw=none},
    ]
    \end{axis}
    \end{tikzpicture}
    \caption{Realization of a Gaussian random field on a unit square domain with Gaussian covariance~\eqref{eq:gaussian-covariance} where $\xi = 0.07$ and $C = 4$ (left) and the same realization transformed by the rescaled sigmoid function $11\rho(x) = 11/(1 + e^{-x})$ (right).}
    \label{fig:sigmoid}
\end{figure}
At the same time, stochastic problems, such as PDEs with random coefficients, are particularly well-suited to quantum algorithms. Because quantum measurements are inherently probabilistic, many quantum routines naturally reduce to estimation tasks. In particular, quantum amplitude estimation provides a quadratic speedup over classical Monte Carlo methods~\cite{Mon15,BHMT02}.

To demonstrate this idea, we study the quantum generation of random fields derived from Gaussian random fields via suitable transformations, together with the estimation of associated scalar quantities of interest, which may be nonlinear functions of the field. A Gaussian random field is a stochastic function on a spatial domain whose pointwise values are normally distributed and whose spatial correlations are prescribed by a covariance function (see Figure~\ref{fig:sigmoid} for a prototypical sample). In many applications, the Gaussian random field is transformed by a nonlinear function, for example, to restrict values to a specific finite range, before being used. We refer to the resulting processes as transformed Gaussian random fields. These are widely used in statistical modeling across many applications, including porous-media flow in geophysics~\cite{ESZ24}, uncertainty quantification for computational mechanics and design optimization~\cite{DKL+24}, and simulations of elastic materials with random microstructure~\cite{BBE+16}. In these settings, the random field typically serves as a coefficient in a PDE, which is the prototypical use case motivating our work.

To the best of our knowledge, this is the first attempt to address the generation and estimation of (transformed) Gaussian random fields in the context of quantum computing. Although extensive methods exist in the classical setting, as reviewed, e.g., in~\cite{LLSY19}, none have been adapted to take advantage of quantum algorithms. Instead of relying on the most widely used classical schemes, we start from a specific technique: the method of moving averages \cite{Oli95}. We show that a slight modification of its standard discretization yields an approximation of the Gaussian random field for which \emph{pointwise} evaluation of realizations is extremely efficient, especially when the covariance function is highly regular. Although isolated pointwise evaluation is typically irrelevant in classical applications, where computations are shared across multiple points, it becomes crucial in the quantum setting as a means to fully exploit quantum parallelism.

Concretely, our method is based on the following construction. Given any classically computable function~$\calA \colon \{0, \dots, 2^{n}-1\} \to [-1, 1]$, we can define a quantum algorithm~$U_\calA$ that performs a diagonal scaling on the computational basis. Specifically, for each basis state~$\ket{j}$, the action is given by
\begin{equation} \label{eq:diagonal-scaling}
U_\calA\ket{j}\ket{0} = \calA(j)\ket{j}\ket{0} + \sqrt{1 - \calA(j)^2}\ket{j}\ket{1},
\end{equation}
as formalized in Lemma~\ref{lem:amplitude-encoding}. Quantum advantage then arises from applying $U_\calA$ to a superposition of basis states, enabling the simultaneous evaluation for all possible inputs. In our case, these inputs encode a random seed specifying the realization of the random field, as well as the spatial point of evaluation.
By integrating this construction with quantum amplitude estimation and a quantum implementation of a classical pseudorandom number generator, an approach previously explored in financial modeling~\cite{MS20, KMTY21}, we obtain a procedure for estimating quantities of interest with complexity that scales nearly linearly with the inverse of the error tolerance. The same framework extends naturally to mixed and higher-order moments of the field. We illustrate and validate our theoretical results through numerical experiments on a simulated quantum device.

\paragraph{Notation}
Random variables and random fields are denoted in boldface (e.g., $\rv{Y}$); 
for a random field $\rv{Y} \colon \Omega \times D \to \mathbb{R}$ 
we write $\rv{Y}(x)$ for the scalar random variable 
$\omega \mapsto \mathbf{Y}(\omega, x)$, suppressing the explicit dependence on $\omega$. Classical objects/functions are denoted in calligraphic letters (e.g., $\mathcal A$), and quantum algorithms (unitaries) by $U$ with appropriate subscripts.  
We also adopt the following notational conventions:\\[0.5ex]
\begin{tabular}{ll}
$[N]$ & The set of integers $0, \dots, N-1$ \\
$\ket{j}_n$ & An $n$-qubit quantum basis state encoding the integer $j \in [2^n]$ \\
$\distN(\mu,\sigma^2)$ & Normal (Gaussian) distribution with mean $\mu$ and standard deviation $\sigma$ \\
$\distN(0,\Sigma)$ & Multivariate Gaussian distribution with zero mean and covariance $\Sigma$ \\
$\distB(n)$ & Uniform distribution over the binary (or bit) strings $\{0, 1\}^n$ of length~$n\in\N$\\
$\Id_n$ & Identity operator on $\C^{2^n}$\\
$\calO(f)$ & Asymptotic notation; hidden constant is independent of any variables in $f$ \\
$C(a, b, \dots)$ & Positive constant depending only on $a, b, \dots$ \\
$|x|, |x|_1$ & Euclidean and $1$-norm of a finite dimensional vector~$x \in \R^d$ respectively \\
$\mathcal{L}(\R^d)$ & Set of all Lebesgue-measurable subsets of $\R^d$.
\end{tabular}

\paragraph{Our contributions}
We address the estimation of stochastic quantities of interest associated with a transformed Gaussian random field,
\[
\tgrf \coloneqq \rho \circ \rv{Y} \colon \Omega \times D \to [-1,1],
\]
where $\Omega$ is a probability space, $D$ is a spatial domain, $\rv{Y}$ is a (mean-zero) Gaussian random field fully characterized by its covariance function $c \colon D\times D\to\mathbb R$, and $\rho \colon \mathbb R\to[-1,1]$ is a bounded Lipschitz function (e.g., the sigmoid used in Figure~\ref{fig:sigmoid}). The uniform bound $|\tgrf|\le 1$ is needed to apply \eqref{eq:diagonal-scaling} and encode samples in the amplitudes of a quantum state. For estimating mixed moments of the random field, we achieve an exponential speedup with respect to the number of evaluation points and a quadratic speedup with respect to accuracy over typical classical implementations. Our approach is structured in three stages, each captured by one of the following main theorems.

\begin{theorem}[Classical pointwise evaluation of Gaussian random fields, Section~\ref{sec:generation}]
\label{thm:main1}
Let $d\in\mathbb N$, let $D\subset\mathbb R^d$ be a convex domain, and let $\rv Y\colon\Omega\times D\to\mathbb R$ be a Gaussian random field with a sufficiently smooth covariance kernel (precise assumptions are stated in Lemma~\ref{lem:quadrature}). Then there exists a classical algorithm $\calA(x, \rv b)$ which is randomized in the sense that it has access to an arbitrary number of random bits~$\rv b \sim \distB(\ast)$ and, given a point $x\in D$, outputs an approximation to $\rv Y(x)$. To reduce the expected approximation error in the supremum norm below an arbitrary tolerance $\tol>0$, the algorithm has computational complexity~$\calO(\polylog \tol^{-1})$. (Naturally, the number of needed random bits~$\rv{b}$ is bounded by this complexity.)
\end{theorem}
\noindent For example, the Gaussian (squared-exponential) covariance
\begin{equation} \label{eq:gaussian-covariance}
c(x,y)=C\,\exp\bigl(-|x-y|^2/(2\xi^2)\bigr),
\end{equation}
with characteristic length scale $\xi>0$ and variance~$C = c(0) > 0$ at a single point, satisfies the assumptions of Theorem~\ref{thm:main1}. 

This classical evaluation routine is a key ingredient for our quantum sampling methods for transformed Gaussian random fields.
\begin{theorem}[Quantum sampling of transformed Gaussian random fields, Section~\ref{sec:sampling}] \label{thm:main2}
    Let $\tgrf \colon \Omega \times D \to [-1, 1]$ be a uniformly bounded random field. Let $n \in \N$ and $x_0, \dots, x_{n-1} \in D$ be points at which $\tgrf$ should be evaluated. Let $\calA$ be a randomized classical algorithm that, given $j\in[n]$ and random bits $\rv b \sim \distB(\ast)$, returns  
    \(\calA(j, \rv b) = \tgrf(x_j),\)
    i.e., a sample of $\tgrf$ at $x_j$.
    Then there exists a randomized construction that, given $m \in \N$, yields a quantum algorithm $\rv U^{(m)}$ that computes $2^m$ samples of $\tgrf$ in the sense that
    \begin{equation} \label{eq:sampler}
    \rv U^{(m)} \ket{j}\ket{k}\ket{0}_a = \tgrf^{(k)}(x_j)\ket{j}\ket{k}\ket{0}_a + \sqrt{1 - \tgrf^{(k)}(x_j)^2}\ket{j}\ket{k}\ket{*}_a,
    \end{equation}
    where the number of ancilla bits~$a \in \N$ depends on $\calA$, and the state $\ket{*}_a$ is orthogonal to $\ket{0}_a$. Here, $\tgrf^{(0)}, \dots, \tgrf^{(2^m-1)}$ are independent and identically distributed (i.i.d.)~copies of $\tgrf$. The algorithm $\rv U^{(m)}$ uses $\calA$ twice, and has $\calO(\poly m)$ other operations.
\end{theorem}
\noindent If the points $x_j$ possess sufficient structure, e.g., they are nodes of a (tensor-product) grid, then the complexity of $\calA$ depends only logarithmically on $n$. 

From this, standard techniques allow can be applied to extract linear quantities of interest. Mixed moments of several such scalar random variables can then be estimated using quantum amplitude estimation~\cite{BHMT02} together with the nonlinear quantum operations of~\cite{DP24a}.
\begin{theorem}[Quantum estimation of mixed stochastic moments, Section~\ref{sec:estimation}] \label{thm:main3}
Let $\tgrf \colon \Omega \times D \to [-1, 1]$ be a uniformly bounded random field and $s \in \N$. For $\ell=1,\dots,s$, define linear quantities of interest
\[
\rv \lambda_\ell \coloneqq q_{0}^{(\ell)} \tgrf(x_0) + \dots + q_{n - 1}^{(\ell)} \tgrf(x_{n-1})
\]
where each $q^{(\ell)} \in \R^n$ satisfies $|q^{(\ell)}|_1 = 1$ and can be efficiently prepared (see~\eqref{eq:prepare}). Assume access to a quantum algorithm $\rv U^{(m)}$ that evaluates $\tgrf$ as in \eqref{eq:sampler}. Then, for any tolerance $\tol>0$ and failure probability $\delta>0$, one can estimate $\E[\rv\lambda_1\rv\lambda_2\cdots\rv\lambda_s]$ 
to absolute error~$\tol > 0$ with failure probability at most~$\delta$, using $\calO(s\tol^{-1}\log \delta^{-1})$ applications of $\rv U^{(m)}$, where $m \in \calO(\log \tol^{-1} \log\log\delta^{-1})$.
\end{theorem}

Combining Theorems~\ref{thm:main1}--\ref{thm:main3}, the total complexity for computing moments of a linear quantity of interest of the transformed random field is
$\calO(\tol^{-1}\log n\polylog \tol^{-1})$. This yields a quadratic improvement in the dependence on $\tol^{-1}$ over standard Monte Carlo, which could as well be achieved classically via quasi-Monte Carlo. However, more importantly, there is an exponential improvement with respect to the number of points~$n$, which is beyond the reach of known classical algorithms. 

\section{Generation of Gaussian random fields}
\label{sec:generation}
In this section, we adapt the classical method of moving averages to optimize its performance for use on a quantum computer. The criteria for efficiently generating Gaussian random fields differ markedly between classical and quantum settings. Classical algorithms typically produce an entire discretized realization at once, whereas a na\"ive use of quantum parallelism demands that the computation at each spatial point be fully independent. We later argue that more sophisticated quantum strategies are unlikely to circumvent this requirement.

The method of moving averages represents the Gaussian random field as a convolution over white noise, where the convolution kernel depends on the covariance function. By approximating the convolution using the trapezoidal quadrature on a Cartesian grid, the white noise is effectively discretized. If the kernel is sufficiently regular, the required grid size scales only logarithmically in the target accuracy, making this an efficient algorithm for point evaluation of the Gaussian random field.

\subsection{Preliminaries on Gaussian random fields}
Gaussian random fields, also called Gaussian processes or Gaussian random functions, model spatially correlated noise.
They are by now well understood; see, e.g., \cite{Lif95,TA07} for thorough treatments. A formal definition is given below.

\begin{definition}[Gaussian random field, covariance function, stationarity]
Let $D$ be any set and $\Omega$ a probability space. A centered \emph{Gaussian random field} is a function $\rv{Y} \colon \Omega \times D \to \R$ such that, for all finite subsets of points~$\{x_1, \dots, x_n\} \subset D$, the distributions of~$(\rv{Y}(x_1), \dots, \rv{Y}(x_n))$ are multivariate Gaussian with zero mean
\[
(\rv{Y}(x_1), \dots, \rv{Y}(x_n)) \sim \distN\big(0, (c(x_i, x_j))_{i, j = 1}^n\big).
\]
The random field is uniquely characterized by its \emph{covariance function}~$c \colon D \times D \to \R_0^+$, defined as 
\[
c(x, y) \coloneqq \Cov(\rv{Y}(x), \rv{Y}(y)).
\]
If $D \subset \R^d$ for some dimension $d \in \N$ and $c$ only depends on the difference~$x - y$, then the field is called \emph{stationary}, and we write
\(
c(x - y) = c(x, y).
\)
\end{definition}

Note that, while in this paper we restrict ourselves to centered Gaussian random fields $\rv{Y}$, this is without loss of generality, as any non-centered Gaussian random field can be centered by subtracting its mean: $\rv{Y} - \E[\rv{Y}]$. Gaussian random fields are particularly useful for statistical modeling, as their distribution is fully determined by the covariance function; see, e.g., \cite[Section~1.2]{TA07}.

In certain cases, it makes sense to generalize the definition of a Gaussian random field to allow for infinite variance at each point evaluation. In other words, we may consider the formal covariance function
\[
c(x, y) = \delta(x - y),
\]
where $\delta$ denotes the Dirac delta function. This corresponds to a limiting case of uncorrelated noise with infinite variance at individual points. A more rigorous formulation is given by the following result.
\begin{lemma}[White noise {\cite[Section~5.6]{Lif95}}]\label{lem:white}
    For any dimension $d \in \N$ there is a random function $\wn \colon \Omega \times \mathcal{L}(\R^d) \to \R$, so-called \emph{white noise}, which satisfies
    \[
    \Cov(\wn(X), \wn(Y)) = \lambda(X \cap Y) \qquad\text{for all } X, Y \in \mathcal{L}(\R^d),
    \]
    where $\lambda$ denotes the Lebesgue measure. It induces a linear integral operator 
    \[
    L^2(\R^d) \ni f\mapsto \int_{\R^d} f(s) \dW \sim \distN(0, \|f\|_{L^2}^2).
    \]
\end{lemma}
\noindent Importantly, white noise serves as a fundamental building block in the method of moving averages, where it is convolved with a kernel function.

\subsection{The method of moving averages} \label{ssec:gaussian-random-fields}
We now discuss how to generate samples of a Gaussian random field given its covariance function. There are many approaches to this problem; see the review~\cite{LLSY19}. Most of these methods are based on computing some form of square root of the covariance kernel. Typically, they begin by discretizing the domain, yielding a finite set $D_h \subset D$. This reduces the problem to finding a square root of the resulting covariance matrix, which can be done, for example, via Cholesky decomposition or using circulant embedding combined with the fast Fourier transform (FFT).

These algorithms generate samples whose distribution matches the restriction of the true Gaussian random field to the discrete set $D_h$ exactly. However, in many practical applications, the number of discretization points $n = |D_h|$ can become very large. Since the kernel decomposition is a global operation, even FFT-based methods or sparse Cholesky decompositions require at least $\mathcal{O}(n)$ operations to evaluate a sample at a single point~$x \in D_h$. The efficiency of these classical methods lies in the fact that evaluating the sample at all $n$ points incurs roughly the same cost as evaluating just one, since the dominant cost is the kernel decomposition itself.

The nonlocal mutual dependence of the point evaluations is not a straightforward use case for quantum parallelism. Instead, we consider a construction that does not require discretization of $D$. A mathematically simple and elegant method that operates in this continuous setting is the method of moving averages~\cite{Oli95}, in which the random field is expressed as the convolution of white noise with a kernel function.

\begin{lemma}[Moving averages] \label{lem:moving-averages}
    Let $D$ be a set and $f \colon D \times \R^d \to \R$ be a function that is $L^2$-integrable with respect to its second argument, i.e.~$f \colon D \to L^2(\R^d)$.
    Further let $\wn$ be white noise on $\R^d$. Then the function $\rv{Y} \colon \Omega \times D \to \R$ defined by
    \[
    \rv{Y}(x) = \int_{\R^d} f(x, s)\dW
    \]
    is a Gaussian random field with covariance
    \[
    c(x, y) \coloneqq  \int_{\R^d} f(x, s)f(y, s) \,\mathrm{d}s.
    \]
\end{lemma}
\noindent The proof can be found in \cite{Oli95}. It is based on the mutual independence of the values of the white noise~$\wn$ at different points. 

A natural question is how to determine a suitable kernel $f$ given a desired covariance function~$c$. This is also discussed in~\cite{Oli95}: if $c$ is stationary, then we can search for a kernel~$f$ such that it depends only on the difference $x - s$. In this case, the random field 
\[
\rv{Y}(x) = f \ast \wn \coloneqq \int_{\R^d} f(x - s) \dW.
\]
is the convolution of $f$ and white noise.
The relation between $c$ and $f$ can then be expressed as $c = f^T \ast f$ where $f^T(x) = f(-x)$. To compute $f$ from a given $c$, we may use the convolution theorem, which states that for suitable functions $f$ and $g$, 
\[
\QFT(f \ast g) = \QFT(f)\QFT(g)
\]
where $\QFT$ is the Fourier transform, using the convention that the factor $(2\pi)^{-d}$ is only multiplied for the inverse transform. 
Applying this to $c = f^T \ast f$, and noting that $\mathcal{F}(f^T) = \overline{\mathcal{F}(f)}$, we find that a valid choice for $f$ is
\[
f = \QFT^{-1}\Big(\sqrt{\QFT(c)}\Big).
\]

\begin{example}[Gaussian covariance]\label{ex:gausscov}
As a concrete example, consider the Gaussian covariance function~\eqref{eq:gaussian-covariance}. A corresponding convolutional square root~$f$ is given by
\begin{equation} \label{eq:gaussian-kernel}
f(x) = \sqrt{C}(\xi^2\pi/2)^{-d/4}\exp\big(-|x|^2/\xi^2\big).
\end{equation}
We refer to~\cite{Oli95} for details. Expressions for kernels corresponding to other common covariance functions can also be found therein.
\end{example}

\begin{remark}[Choice of algorithm and quantum limitation]
It may seem natural to base our method on the discrete Fourier transform, since its quantum analogue can be implemented very efficiently. Doing so, however, would require encoding point evaluations $\rv{Y}(x)$, which are Gaussian and therefore unbounded, into the amplitudes of a quantum state. Because quantum states must be normalized, no fixed global rescaling can map an unbounded random variable to amplitudes. Even if one truncates the values of $\rv{Y}(x)$ to achieve an error tolerance $\varepsilon$, the required rescaling factor is proportional to $\varepsilon^{-1/2}$, which leads to a corresponding loss in efficiency. For this reason, we adopt an approach that avoids amplitude encoding of unbounded Gaussian samples.
\end{remark}

\subsection{Discretization of the Gaussian Random Field} \label{ssec:gaussian-discretization}
We depart again from the stationary setting, where there is not necessarily any relation between the variables $x$ and $s$ of the kernel~$f$. Still, in a conventional implementation of the method of moving averages, the domain~$D$ is typically identified with the domain~$\R^d$ of the convolution variable~$s$, and both are discretized in the same way. We adopt a different strategy and discretize only this integral while leaving $D$ continuous. We fix the convolution kernel~$f$, the white noise~$\wn$, and the resulting Gaussian random field~$\rv{Y}$ constructed from them as in Lemma~\ref{lem:moving-averages}.

To discretize the integral in the moving average formulation, we consider a $\sinc$ interpolation of the kernel~$f$ with grid size~$h > 0$, defined by
\[
f_\rh(x, s) \coloneqq \sum_{j \in \Z_r^d} f(x, jh) \sinc\left(\frac{s}{h} - j\right),
\]
where we restrict the sum to a finite grid $\Z_r^d \coloneqq \Z^d \cap B_r(0)$, with $B_r(0)$ denoting the (closed) Euclidean ball of radius $r$ centered at the origin.
Here, the multidimensional $\sinc$ function is given by
\[
\sinc(s) = \prod_{k=1}^d \frac{\sin(s_k\pi)}{s_k\pi}.
\]
It satisfies $\sinc(0) = 1$ and $\sinc(j) = 0$ for all $j \in \Z^d\setminus\{0\}$, so that $f_\rh(x,\bullet)$ interpolates $f(x,\bullet)$ at the grid points, that is, for all $x\in D$, 
$$f_\rh(x, jh) = f(x, jh), \quad j\in\Z_r^d.$$
Moreover, the set of $\sinc$ functions shifted by integer values is $L^2$-orthonormal, meaning
\[
\int_{\R^d} \sinc(s - j)\sinc(s - k) \,\mathrm{d}s = 0 \qquad \text{for all } j, k \in \Z^d \text{ with } j \neq k.
\]

Using Lemma~\ref{lem:moving-averages} with the interpolated kernel~$f_\rh$ gives the discretized random field
\begin{equation} \label{eq:discrete-random-field}
\rv{Y}_\rh(x) = \int_{\R^d} f_\rh(x, s) \dW = h^{d/2} \sum_{j \in \Z_r^d} f(x, jh) \wn_j,
\end{equation}
where, using the orthogonality of the shifted $\sinc$ functions, the random variables
\begin{equation} \label{eq:white-noise-sinc}
\wn_j \coloneqq h^{-d/2}\int_{\R^d} \sinc\left(\frac{s}{h} - j\right) \dW \sim \distN(0, 1)
\end{equation}
are independent and standard normally distributed. Since in this way $\rv{Y}$ and its approximation $\rv{Y}_\rh$ are defined in terms of the same random variable $\rv{W}$, they can be directly compared for each sample, rather than requiring a more complicated error measure such as a Wasserstein distance. In practice, however, the i.i.d.\ standard normal variables $\wn_j$ can be sampled directly, without evaluating the integral in~\eqref{eq:white-noise-sinc}. Computing a point evaluation of $\rv{Y}_\rh$ using~\eqref{eq:discrete-random-field} then takes $\calO(r^d)$ operations, including the sampling of these variables.

We will show that, by choosing the discretization parameter according to $h = r^{-1/2}$, the approximation $\rv{Y}_\rh$ converges exponentially to the true field~$\rv{Y}$ as a function of the truncation radius~$r$. To see this, consider the covariance function
\[
c_\rh(x, y) = h^d\sum_{j \in \Z_r^d} f(x, jh)f(y, jh)
\]
of $\rv{Y}_\rh$. This expression corresponds to applying the trapezoidal rule to approximate the integral in the covariance expression from Lemma~\ref{lem:moving-averages}. As is well known, the trapezoidal quadrature rule exhibits exponential convergence with respect to the mesh size~$h$, provided the integrand is sufficiently smooth and decays rapidly at infinity. 
\begin{lemma}[Exponentially convergent trapezoidal rule {\cite[Theorem~6.1]{TW14}}] \label{lem:trapezoid}
    Let $g \colon \R^d \to \R$ and $\alpha_g, \beta_g, \gamma_g > 0$ be constants such that $g$ extends to an analytic function
    \[
    \tilde g \colon (\R + (-\alpha_g, \alpha_g)i)^d \to \C,
    \]
    which decays exponentially in the sense that
    \begin{equation}
    \label{e:gdecay}    |\tilde g(s)| \le \gamma_g e^{-\beta_g|\Re s|}\quad\text{for all } s \in (\R + (-\alpha_g, \alpha_g)i)^d.
    \end{equation}
    Then there is a constant $C=C(\gamma_g,\beta_g,d)>0$ independent of $h$ and $r$ such that the trapezoidal approximation over the truncated  grid $\Z_r^d$ fulfills 
    \[
    \bigg| \int_{\R^d} g(s) \,\mathrm{d}s - h^d\sum_{j \in \Z_r^d} g(jh)\bigg| \leq C(e^{-2\pi \alpha_g/h} + e^{-\beta_ghr/2}).
    \]
    The constant~$C \approx (2d)^d$, however, may critically depend on the spatial dimension. 
\end{lemma}
\begin{proof}
Proofs of the error bound $2\gamma_g/\beta_g e^{-2\pi \alpha_g / h}$ for the trapezoidal rule on the real line can be found in~\cite[Theorem 5.1]{TW14}. This bound remains valid in the multidimensional setting, up to a different dimension-dependent multiplicative constant; see Appendix~\ref{sec:multiple-dimensions-proof} for details.  It remains to estimate the additional contribution to the error from truncating the quadrature to the finite index set $\Z_r^d$. Using the decay property~\eqref{e:gdecay}, we obtain 
\begin{align*}
\biggl| h^d \;\sum_{\mathclap{j \in \Z^d \setminus \Z_r^d}}\;\; g(jh) \biggr| &\le \gamma_gh^d \;\sum_{\mathclap{j \in \Z^d \setminus \Z_r^d}}\;\; e^{-\beta_g|jh|}
\le \gamma_g\int_{\mathrlap{\R^d \setminus B_{rh}(0)}}\qquad e^{-\beta_g|x|} \dx = \frac{S_{d-1}\gamma_g}{\beta_g}\int_{\beta_ghr}^\infty\! x^{d-1}e^{-x} \dx \\
&\le \frac{S_{d-1}C_d\gamma_g}{\beta_g}\int_{\beta_ghr}^\infty\! e^{-x/2} \dx = \frac{2S_{d-1}C_d\gamma_g}{\beta_g} e^{-\beta_ghr/2}, 
\end{align*}
where $S_{d-1}$ is the surface volume of the $(d-1)$-dimensional unit sphere and $C_d\leq(2d + e^2)^{d-1}$ is the maximum of $x^{d-1}e^{-x/2}$ for $x\geq 0$.
\end{proof}
In cases where $g$ admits an analytic extension to the entire complex space \( \mathbb{C}^d \), as is the case for the Gaussian kernel in~\eqref{eq:gaussian-kernel}, a sharper error bound can be achieved by allowing the parameter \( \alpha_g \) in Lemma~\ref{lem:trapezoid} to vary with \( h \); see~\cite{Goo49, TW14} for details.  

By applying Lemma~\ref{lem:trapezoid} to the function $g(s) = f(x, s)f(y, s)$, one finds that the discretized covariance $c_h$ closely approximates the true covariance $c$. Moreover, the error in approximating the Gaussian random field can be explicitly bounded even in the supremum norm.
\begin{lemma}[Error of $\sinc$-interpolation of Gaussian random fields] \label{lem:quadrature}
    Let $D$ be convex, and let $f \colon D\times\R^d\rightarrow \R$ 
    satisfy the following conditions:
    \begin{enumerate}
        \item There are constants $\alpha_f, \beta_f, \gamma_f > 0$ such that, for all $x\in D$, 
        $f(x,\bullet)$ extends to an analytic function 
    \[
    \tilde f(x,\bullet) \colon (\R + (-\alpha_f, \alpha_f)i)^d \to \C
    \]
    which decays exponentially in the sense that, for all $x\in D$, 
    \[
    |\tilde f(x,s)| \le \gamma_f e^{-\beta_f|\Re s|}\quad\text{for all } s \in (\R + (-\alpha_f, \alpha_f)i)^d.
    \]
    (That is, $f(x,\bullet)$ satisfies the assumption of Lemma~\ref{lem:trapezoid} uniformly for all $x$.)
    \item $f$ is differentiable with respect to $x$ and its directional derivatives decay fast, in the sense that there are constants $\gamma_f', \beta_f' > 0$ such that, for all $x \in D$ and all normalized $v\in \R^d$, the directional derivatives $\frac{\mathrm{d}}{\mathrm{d}t} f(x + tv, s)$ satisfy
    \[
    \left|\frac{\mathrm{d}}{\mathrm{d}t} f(x + tv, s)\right| \le \gamma_f' e^{-\beta_f'|s|}\quad\text{for all } s \in \R^d.
    \]
    \end{enumerate}
    Let $\rv{Y}=\int_{\R^d}f(x,s) \dW$ be a centered Gaussian random field and $\rv{Y}_\rh$ its $\sinc$-based approximation  from \eqref{eq:discrete-random-field} with truncation radius $r>0$ and grid size $h_r = r^{-1/2}$.
    Then for any $\tol > 0$ there exists some $r=r(\tol) \in \calO(\log^2(1/\tol))$ such that the error satisfies
    \[\E\left[\sup_{x\in D}|\rv{Y}(x) - \rv{Y}_{r\!,h_r}(x)|\right] \le \tol.\]
\end{lemma}
\begin{proof}
\renewcommand{\rh}{r}
We write $\rv{Y}_r = \rv{Y}_{r\!,h_r}$ and $f_r = f_{r\!,h_r}$, since $h_r$ is fully determined in terms of $r$.
As both random fields are defined in terms of the white noise $\wn$, we obtain an explicit error for each $\omega \in \Omega$, specifically
\[
\rv{E}(x) \coloneqq \rv{Y}(x) - \rv{Y}_\rh(x) = \int_{\R^d} f(x, s) - f_\rh(x, s)  \dW.
\]
Clearly, the error $\rv{E}(x)$ is itself a Gaussian random field with the covariance kernel $f - f_\rh$. Recall from Lemma~\ref{lem:white} that 
$\E[\rv{E}(x)]=0$ and that \begin{align}\nonumber\operatorname{Var}[\rv{E}(x)]=\E[\rv{E}(x)^2]&=\|f(x, s) - f_\rh(x, s) \|^2_{L^2(\R^d)}\\&=\int_{\R^d} f(x, s)^2 \,\mathrm{d}s - 2 \int_{\R^d} f(x, s)f_\rh(x, s) \,\mathrm{d}s + \int_{\R^d} f_\rh(x, s)^2 \,\mathrm{d}s.\label{e:applyL10}\end{align}
Due to the orthogonality of shifted $\sinc$ functions, we observe 
\[
\int_{\R^d} \!f_\rh(x, s)^2 \,\mathrm{d}s = h_r^d\sum_{\mathclap{j,k \in \Z^d_r}} f(x, jh_r)f(x, kh_r) \int_{\R^d} \!\sinc(s - j)\sinc(s - k) \,\mathrm{d}s =h_r^{d} \sum_{j \in \Z^d_r} f(x, jh_r)^2.
\]
The right-hand side is precisely the truncated trapezoidal quadrature used in Lemma~\ref{lem:trapezoid} for the function $g(s)\coloneqq f(x, s)^2$.
It is easily checked that $g$ satisfies the assumption of Lemma~\ref{lem:trapezoid} with constants
$\alpha_{g} = \alpha_f$, $\beta_{g} = 2\beta_f$, and  $\gamma_{g} = \gamma_f^2$. Thus, the right-hand side of \eqref{e:applyL10} is bounded by 
\begin{equation}\label{e:applyL10j}
2\,\biggl|\int_{\R^d} f(x, s)f_\rh(x, s) \,\mathrm{d}s -h_r^d\sum_{j \in \Z^d_r} f(x, jh_r)^2\biggr|
\end{equation}
up to an additive error term $C\left(e^{-2\pi \alpha_f/h_r} + e^{-4\beta_fh_rr}\right)$. Using the definition of $f_\rh$ and introducing $g_j(s) \coloneqq f(x, s)\sinc(s/h_r - j)$ for $j\in \Z^d_r$, the term \eqref{e:applyL10j} is bounded by 
\begin{equation}\label{e:applyL10j2}
2\sum_{j\in\Z^d_r}|f(x, jh_r)|\biggl|\int_{\R^d} g_j(s) \,\mathrm{d}s - h_r^dg_j(jh_r)\biggr|,
\end{equation}
since $g_j(jh_r)=f(x,jh_r)$ and $g_j(kh_r)=0$ for all $k\neq j$. This again involves trapezoidal quadrature errors for the integrals of the $g_j$ (without truncation). The $g_j$ satisfy the assumptions of Lemma~\ref{lem:trapezoid} with constants $\alpha_{g_j} = \alpha_f$, $\beta_{g_j} = \beta_f$, and $\gamma_{g_j} = 2e^{\pi \alpha_f}\gamma_f$,
because the $\sinc$ function is analytic on the whole complex plane,  and $\lvert\sinc(z)| \le 2e^{\pi\lvert\Im z|}$ for all $z\in \R + (-\alpha_f, \alpha_f)i \subset \C$; see Appendix~\ref{sec:sinc-bound}.
Thus, Lemma~\ref{lem:trapezoid} yields 
$$\biggl|\int_{\R^d} g_j(s) \,\mathrm{d}s - h_r^dg_j(jh_r)\biggr|\leq C(\gamma_f,\alpha_f,d)e^{-2\pi \alpha_f/h_r}.$$
As a consequence, the sum~\eqref{e:applyL10j2} can then be bounded by
\begin{align*}
\sum_{j \in \Z^d_r} |f(x, jh_r)| C(\gamma_f,\alpha_f,d)e^{-2\pi \alpha_f/h_r} &\le h_r^{-d} C(\gamma_f,\alpha_f,d)e^{-2\pi \alpha_f/h_r} h_r^d \sum_{j \in \Z^d_r} \gamma_f e^{-|\beta_f|} \\
&\le C(\gamma_f, \alpha_f, \beta_f, d) e^{-\pi \alpha_f/h_r}.
\end{align*}
Note that the term $h_r^{-d}$ has been absorbed by reducing the exponential decay rate and by adjusting the multiplicative constant, which may in any case grow exponentially with the spatial dimension~$d$. Altogether, we achieve the error bound
\begin{equation} \label{eq:error-variance}
\Var[\rv{E}(x)]\leq C(\gamma_f, \alpha_f, \beta_f, d) (e^{-\pi \alpha_f/h_r} + e^{-4\beta_fh_rr}) \le e^{-b_{\Var}\sqrt{r}}
\end{equation}
with $C_{\Var} = C(\gamma_f, \alpha_f, \beta_f, d)$ and $b_{\Var} = C(\alpha_f, \beta_f)$, where we used $h_r = r^{-1/2}$.

To show the boundedness of this random field, we now apply~\cite[Theorem~1.3.3]{TA07}, which gives the following entropy integral bound
\begin{equation} \label{eq:entropy-integral}
\E\left[\sup_{x\in D} |\rv{E}(x)|\right] \le \int_0^\infty (\log N(\delta))^{1/2} \,\mathrm{d}\delta.
\end{equation}
Here, $N(\delta)$ denotes the minimal number of balls of radius~$\delta$ (in the pseudo-metric induced by the error field~$\rv{E}$) required to cover the domain~$D$. These balls are defined as
\[
B_\delta(x) \coloneqq \{ y \in D \mid \E[(\rv{E}(x) - \rv{E}(y))^2] \le \delta^2 \}.
\]
Irrespective of the distance between $x$ and $y$, we can estimate
\begin{align*}
\E\big[(\rv{E}(x) - \rv{E}(y))^2\big] \le 2\E\big[\rv{E}(x)^2\big] + 2\E\big[\rv{E}(y)^2\big]&=2\Var[\rv{E}(x)]+2\Var[\rv{E}(y)]\\&\leq 2C_{\Var} e^{-b_{\Var}\sqrt{r}} \eqqcolon \mu(r)^2
\end{align*}
using \eqref{eq:error-variance}. This means for $\delta\ge\mu(r)$ we have $N(\delta) = 1$. 
However, to evaluate the integral~\eqref{eq:entropy-integral}, we also require a bound as $x - y \to 0$. For any $v \in \R^d$ with $|v| = 1$, we thus consider $y = x + tv$ where $t \to 0$. The mean value theorem implies that there is some $0 < \tilde t < t$ such that, together with assumption (2),
\[
|f(x + tv, s) - f(x, s)| \le t |\tfrac{\mathrm{d}}{\mathrm{d}t} f(x + \tilde t v, s)| \le \gamma_f' t e^{-\beta_f'|s|}.
\]
In this regime, we get
\begin{align*}
&\E[(\rv{E}(x) - \rv{E}(x + tv))^2] = \int_{\R^d} \big((f(x, s) - f(x + tv, s)) - (f_\rh(x, s) - f_\rh(x + tv, s))\big)^2 \,\mathrm{d}s \\
&\qquad\le 2 \int_{\R^d} (f(x, s) - f(x + tv, s))^2 + \bigg(\sum_{j \in \Z_r^d} (f(x, s) - f(x + tv, s)) \sinc\left(\frac{s}{h_r} - j\right) \bigg)^2 \,\mathrm{d}s \\
&\qquad= 2 t^2(\gamma_f')^2\bigg(\int_{\R^d}\! e^{-2\beta_f|s|} \,\mathrm{d}s + h_r^d\sum_{j \in \Z_r^d} e^{-2\beta_f|s|}\bigg) \le 2t^2 S_{d-1}C_d (\gamma_f')^2 / \beta_f'.
\end{align*}
using an analogous calculation as in the proof of Lemma~\ref{lem:trapezoid}. Defining $M \coloneqq S_{d-1}C_d \gamma_f'^2/\beta_f'$ we get
\[
N(\delta) \le \tilde N(\delta / M),
\]
where $\tilde N$ is defined equivalently to $N$ but for balls in the Euclidean metric, that is, 
\[
\tilde N(\delta) \le (C_N \operatorname{diam}(D)/\delta)^d
\]
for some universal constant $C_N$. Combining this with the previous bound for $N$, we conclude
\[
N(\delta) \le \begin{cases}
1 & \delta \ge \mu(r) \\
(C'/\delta)^d & \text{otherwise}
\end{cases}
\]
with $C' \coloneqq MC_N \operatorname{diam}(D)$.
Together with~\eqref{eq:entropy-integral}, this yields
\begin{align*}
\E\left[\sup_{x\in D} |\rv{E}(x)|\right] &= \sqrt{d}\int_0^{\mu(r)} (\log(C'/\delta))^{1/2}  \,\mathrm{d}\delta\\ &= \sqrt{d}\left(\mu(r)\sqrt{\log(C'/\mu(r))} + \tfrac12 \sqrt{\pi} C' \operatorname{erfc}(\log(C'/\mu(r)))\right) \\
&\le \sqrt{dC_{\Var}}e^{-b_{\Var} \sqrt{r}/2}\left(\sqrt{\log(C') + b_{\Var} \sqrt{r}} + \tfrac{1}{2}\sqrt{\pi}\right) 
\le C(d, \operatorname{diam}(D), f) e^{-b_{\Var} \sqrt{r}/2}
\end{align*}
where we used $\operatorname{erfc}(x) \le e^{-x}$. This proves the assertion, noting that the upper bound decays exponentially with $\sqrt{r}$.
\end{proof}

\begin{remark}[Dependence on correlation length] \label{rem:correlation-length}
Let $\rv{Y} \colon \Omega \times D \to \R$ be a random field on $D \subset \R^d$. Then we can consider the random field $\rv{Y}^{(\xi)} \coloneqq \rv{Y}(\xi^{-1} \bullet)$, where the spatial coordinate is scaled down by a factor~$0 < \xi < 1$. The convolution kernel for this random field is explicitly given by $f(\xi^{-1} x, s)$ or equivalently $\xi^{-d/2}f(\xi^{-1} x, \xi^{-1} s)$, where the latter preserves stationarity of the kernel. If, for example, $\rv{Y}$ is produced using the Gaussian convolution kernel~\eqref{eq:gaussian-kernel} with characteristic length~$1$, then $\rv{Y}^{(\xi)}$ can be obtained by the same kernel with characteristic length~$\xi$. 

The error measure of Lemma~\ref{lem:quadrature} is preserved under this spatial rescaling.
However, the domain~$\xi D$ of $\rv{Y}^{(\xi)}$ is now smaller by the factor~$\xi$. Thus, varying the correlation length in Lemma~\ref{lem:quadrature} is equivalent to varying the size of the domain~$D$. The diameter of $D$ only enters logarithmically through the constant $C'$, but there may be additional implicit dependencies through the constants $\beta_f$ and $\gamma_f$. The latter can be avoided for stationary covariance kernels, i.e.~$f(x, s) = f(x - s)$, by choosing a slightly different truncation of the sum than~\eqref{eq:discrete-random-field}, centered around $x$, namely
\[
\rv{Y}_\rh(x) = h^{d/2} \;\sum_{\mathclap{j \in \Z^d_r + \lfloor x / h \rfloor}}\quad f(x - hj) \rv{W}_j.
\]
Note that, in any case, the scaling argument applied in the construction of $\rv{Y}^{(\xi)}$ applies here as well, and one can calculate $\rv{Y}_\rh(\xi^{-1} \bullet) = \rv{Y}^{(\xi)}_{r\!,\xi h}$, which indicates that the grid size~$h$ should scale linearly with the correlation length.
\end{remark}

\begin{remark}[Error in $L^p$ norm]
A simpler proof than that of Lemma~\ref{lem:quadrature} is possible for bounding the error in an $L^p$ norm for $1 \le p < \infty$. By combining~\cite[Equation~3.32]{Hen24} with H\"older's inequality, this error is bounded by
\[
\E[\|\rv{Y} - \rv{Y}_\rh\|_{L^p}] \le C_p \,\big\lVert\Var[\rv Y(\bullet) - \rv Y_\rh(\bullet)]^{1/2}\big\rVert_{L^p},
\]
where $C_p > 0$ is a constant depending only on $p$.
\end{remark}

Theorem~\ref{thm:main1} follows as a corollary of Lemma~\ref{lem:quadrature} using the complexity $\calO(r^d)$ of sampling $\rv{Y}_\rh$ with~\eqref{eq:discrete-random-field}.
Compared to existing methods in the literature, our approach significantly reduces the number of required quadrature points for smooth kernels. This reduction is often negligible in classical applications, where the number of evaluation points $x_0, \dots, x_{n-1}$ dominates the complexity instead. However, it becomes essential in the quantum setting. To realize the exponential advantage of quantum parallelism, described in~\eqref{eq:diagonal-scaling}, the computations at each point need to be fully independent and intermediate results cannot be reused. This independence is a key structural requirement that we will exploit in the following sections.

\section{Quantum sampling of random fields}
\label{sec:sampling}

The sampling method from the previous section yields a classical algorithm
\begin{equation} \label{eq:vec-algorithm}
\mathcal{A} \colon \R^d \times \R^{r^d} \to \R, \quad \text{such that} \quad \mathcal{A}(x, (\wn_k)_{k = 1}^{r^d}) \sim \rv{Y}(x),
\end{equation}
which produces samples from the Gaussian random field at the evaluation point~$x$. Recall that the randomness is introduced via a small number of independent standard normal variables~$\wn_k$.
This approximation is efficient in the sense that the number of random variables~$r^d$ scales only polylogarithmically with the approximation error.
Still, we further reduce the dependence on the number of random variables through the use of \emph{seekable pseudorandom number generators}, in order to use quantum parallelism as in~\eqref{eq:diagonal-scaling} to evaluate $\calA$ for many inputs and samples at once.

\subsection{Seekable pseudorandom number generators}
Random bits serve as the foundation for sampling from arbitrary distributions. To sample from a target distribution, we may concatenate $\Nseed$ random bits into an integer $j \in [2^{\Nseed}]$, and then compute the value
\begin{equation} \label{eq:pdf}
\Phi^{-1}(2^{-\Nseed} j)
\end{equation}
where $\Phi^{-1}$ denotes the inverse cumulative distribution function (CDF) of the desired distribution. For the normal distribution, the Box-Muller method~\cite{BM58} gives accurate samples in a similar fashion.

While classical computers have to rely on pseudorandomness for sampling the random bits for $j$, quantum computers can, in principle, generate truly random bits by measuring a qubit in the state
\[
\ket{+} = \mathrm{H}\ket{0} = (\ket{0} + \ket{1})/\sqrt{2}.
\]
However, this method requires one qubit for each random bit, which becomes inefficient when many random numbers are needed, as in the case of sampling a random field. The inefficiency becomes especially pronounced when the correlation length~$\xi$ becomes small as considered in Remark~\ref{rem:correlation-length}. Although the number of bits required per point of the random field remains bounded independently of~$\xi$, the total number of bits needed across all points grows polynomially with~$\xi^{-1}$. This creates the need for a more efficient approach for generating random numbers in the quantum setting.

To address this, similar to classical methods, we turn to \emph{pseudorandom number generators} (PRNGs). These are functions that, given a random seed, produce output sequences that appear statistically random. One class of PRNGs particularly well-suited to our application is the family of permuted congruential generators (PCGs)~\cite{ON14}. PCGs are efficient in terms of state size and simple to implement on quantum hardware. Moreover, they allow for efficient seeking to arbitrary positions in the sequence. This is a key requirement for quantum parallelism, as noted earlier in~\cite{KMTY21, MS20}, where PCGs were already used in the quantum setting.
We now formalize the interface required of a pseudorandom number generator for our purposes.

\begin{definition}[Seekable pseudorandom number generator]
A \emph{seekable pseudorandom number generator} is a classical algorithm~$\RNG$ that takes as input a \emph{state size}~$\Nseed$, a \emph{seed}~$s \in [2^\Nseed]$, and an index~$j \in [2^M]$ and outputs a bit~$\RNG(s, j) \in \{0, 1\}$. We require the runtime of $\RNG$ to be $\calO(\poly \Nseed)$.
\end{definition}

The algorithm should produce seemingly random bits; that is,
$\text{``}\RNG(\rv{s}, \bullet) \sim \distB(2^M)\text{''}$,
where $\rv{s} \sim \distB(\Nseed)$. However, this is not achievable in the strict sense, since $\RNG$ is deterministic by design. Clearly, the number of possible output sequences is larger than the number of possible seeds. 
A commonly accepted middle ground are \emph{cryptographically secure} pseudorandom number generators. These generators do not produce truly random sequences, but their output is computationally indistinguishable from true randomness. In other words, it is hard (relative to the state size~$M$) to algorithmically distinguish between a pseudorandom sequence and a truly random one.

\begin{definition}[Cryptographically secure random generator {\cite[Definition~3.14]{KL14}}] \label{def:secure-prng}
    A pseudorandom number generator~$\RNG$ is called \emph{asymptotically secure}, if no algorithm~$\mathcal{D}$ with polynomial runtime in the state size~$M \in \N$ can efficiently distinguish between the pseudorandom sequence and a truly random sequence $\rv b \in \distB([2^M])$. More precisely, the value
    \[
    |\Pr(\mathcal{D}(\RNG(\rv{s}, \bullet)) = 1) - \Pr(\mathcal{D}(\rv b) = 1)|,
    \]
    where $\rv{s} \sim \distB(\Nseed)$, should be negligible with respect to $M$. See \cite{KL14} for the quantitative definition of ``negligible''.
\end{definition}

While secure pseudorandom generators can be constructed under relatively mild assumptions; see~\cite[Section~7.2]{KL14}, the PCGs we suggest are not believed to be asymptotically secure. In principle, it may be possible to construct a \emph{malicious} quantity of interest, specifically designed to exploit statistical weaknesses in the PCG, leading to poor approximation of the random field. Despite this, we use PCGs for the numerical experiments due to their practical efficiency and simplicity, as is common also in classical applications. This is justified further by standardized statistical tests which PCGs pass; see~\cite{ON14}. Our theoretical results, however, are stated under the assumption that a cryptographically secure pseudorandom number generator is used.

\subsection{Quantum sampling by diagonal scaling}
We now combine a random generator~$\RNG$ with the algorithm~$\calA$ from~\eqref{eq:vec-algorithm} for point evaluation of the random field, as well as a structured way to compute evaluation points $x_0, \dots, x_{n-1} \in D$, such as placing them at nodes of a regular grid. This yields an algorithm that computes
\[
[n] \times [M] \to \R, \qquad (j, \rv{s}) \mapsto \rv{Y}(x_j),
\]
where $\rv{s} \sim \distB(\Nseed)$. Strictly speaking, the distribution of the values $\rv{Y}(x_j)$ generated in this way may differ from the ideal case in~\eqref{eq:vec-algorithm}, since $\RNG$ is deterministic. However, if the random generator is cryptographically secure, the distributions are computationally indistinguishable.

We aim to prepare quantum states in which samples of the random field are encoded in the amplitudes, rather than in binary registers. A natural target is the quantum state
\[
\rv{Y}(x_0) \ket{0} + \dots + \rv{Y}(x_{n-1})\ket{n-1}.
\]

However, the amplitudes of quantum states cannot exceed $1$. Since the values of $\rv{Y}$ are normally distributed and therefore unbounded, this constraint cannot be satisfied by applying a fixed linear rescaling. 
Instead, we could ensure such a bound, for example, by applying some bounded non-linear transformation function~$\rho \colon \R \to [-1, 1]$ to its values, i.e.\ we consider
\[
\tgrf \coloneqq \rho \circ \rv Y.
\]

We note that the covariance structure of the transformed field~$\tgrf$ is not directly determined by the original covariance~$c$; see~\cite{GHCL13} for further discussion. 
Nevertheless, such transformations of Gaussian random fields, used to impose lower and upper bounds, or even to restrict outputs to finitely many values in classification tasks, are natural in many modeling contexts. One notable example is the generation of two-phase microstructures in heterogeneous materials.

\begin{example}[Generation of two-phase microstructure]
    One possible application for our algorithm is generating properties of a material with a two-phase microstructure, where some quantity takes two values $0 < Z_0 < Z_1$ depending on whether $\rv Y(x)$ is larger or smaller than some cutoff value~$b \in \R$. Specifically, $
    \tgrf(x) = Z_0 + (Z_1 - Z_0) \rho(a\rv Y(x) - b)$, 
    where $a > 0$ determines the thickness of the transition layer, and $\rho$ is the sigmoid function from Figure~\ref{fig:sigmoid}, which also shows an example realization of $\tgrf$ as defined here. Clearly, in this case $Z_0$ and $Z_1$ are lower and upper bounds on the coefficient, ensuring its uniform ellipticity. Information about the microstructure is imposed through the covariance function of $\rv{Y}$. One possibility to encode this data is modeling the covariance as a sum of Gaussian curves, which could be implemented as a sum of independent random fields, see~\cite[Eqn.~23]{Oli95}.
\end{example}

If the transformation $\rho$ is classically computable, we obtain a classical algorithm 
\[
\tilde{\mathcal{A}} \colon [n] \times [2^\Nseed] \to [-1, 1], \quad (j, \rv{s}) \mapsto \tgrf(x_j) = \rho(\rv{Y}(x_j))
\]
that evaluates the transformed random field. We can extend this construction by using the pseudorandom number generator~$\RNG$ to generate multiple seeds from a single one. This yields the algorithm
\[
\tilde{\calA}^{(m)} \colon [n] \times [2^m] \times [2^\Nseed] \to [-1, 1], \qquad (j, k, \rv s) \mapsto \tilde{\calA}(\RNG(\rv s, k), j).
\]
For a fixed seed~$\rv{s}$, the algorithm~$\tilde{\mathcal{A}}^{(m)}$ computes the values of the transformed field $\tgrf = \rho \circ \rv Y$ at the interpolation points~$x_j$ for $2^m$ independent samples.

We now combine~$\tilde{\mathcal{A}}^{(m)}$ with a standard construction from quantum computing to prepare quantum states encoding these samples.

\begin{lemma}[Amplitude encoding of classical algorithms] \label{lem:amplitude-encoding}
Let $n \in \N$ and let $\calC \colon [n] \to [-1, 1]$ be a function that is classically computable to arbitrary precision. Then, for any target accuracy $\tol > 0$, there exists a quantum algorithm~$U_\calC$ such that
\[
U_\calC \ket{0}_1\ket{0}_\anc\ket{j}_{\lceil \log n \rceil} = \calC(j) \ket{0}_1\ket{0}_\anc\ket{j} + \sqrt{1 - \calC(j)^2} \ket{1}_1\ket{0}_\anc\ket{j},
\]
holds up to an error of $\tol > 0$. The required number of ancilla bits~$\anc \in \N$ and the complexity of $U_\calC$ scale linearly with the classical complexity of computing $\calC$ up to precision $\tol$.
\end{lemma}
\begin{proof}
See Figure~\ref{fig:amplitude-encoding} for the circuit implementing $U_\calC$. Given a basis state~$\ket{j}$ it first computes the integer~$\theta = \lfloor 2^m \cos^{-1}(\calC(j)) / \pi \rfloor$ into an additional register~$\ket{k}$ of $m \coloneqq \lceil \log_2(\pi/(1-\cos(\tol)))\rceil$ qubits. This operation can be implemented using a number of gates linear in the number of operations of $\calC$ since it is reversible, see~\cite{Tof80}, potentially using additional space $\anc' \in \N$ and producing garbage data~$\ket{\ast}_{\anc'}$ in this additional space, leading to a state
\[
\ket{0}_1\ket{\ast}_{\anc'}\ket{\theta}_m\ket{j}.
\]
It then applies multiple controlled rotation gates to implement the operation
\[
\ket{\theta}_m \ket{0} \mapsto \ket{\theta}_m \mathrm{R}_\mathrm{Y}(2^{-m+1} \theta \pi) \ket{0} = \ket{\theta}_m (\cos(2^{-m} \theta \pi) \ket{0} + \sin(2^{-m} \theta \pi) \ket{1}).
\]
Afterwards, the registers~$\ket{\theta}$ and $\ket{\ast}_{\anc'}$, which remained unchanged, are ``uncomputed'' again, meaning the inverse circuit for computing $\ket{\theta}$ is applied, resulting in the two registers being returned to their initial state.
The correctness of the algorithm, i.e.~$\cos(2^{-m+1} \theta \pi) \approx \calC(j)$, can easily be checked. The total number of ancilla bits is $\anc = m + \anc'$.
\end{proof}

\begin{figure}
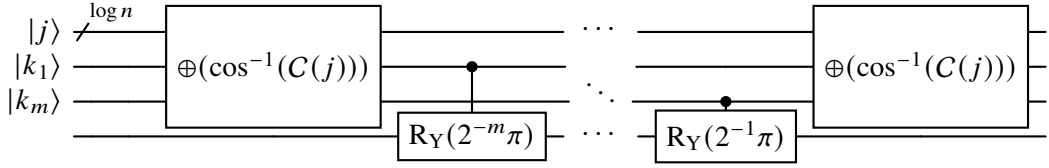

\centering
\includestandalone{circuit_rotation}
\caption{The circuit~$U_\calC$ from Lemma~\ref{lem:amplitude-encoding}. The gate $\oplus(\cos^{-1}(\calC(j)))$ stores its result in a fixed-point encoding as detailed in the proof of the lemma. In our implementation it is composed of many fixed-point arithmetic operations. This circuit does not show the additional space $\anc'$.
}
\label{fig:amplitude-encoding}
\end{figure}

The algorithm $U_\calC$ corresponds to a diagonal scaling of the basis states~$\ket{j}$ by $\calC(j)$.
The algorithm~$\rv{U}^{(m)}$ from Theorem~\ref{thm:main2} is obtained by $\calC = \tilde{\calA}^{(m)}(\bullet, \bullet, \rv s)$,
where $\rv s$ is sampled classically. A minor technical point is that the output of $\tilde{\calA}^{(m)}$ is not exactly distributed as the values of $\tgrf$ due to the use of a pseudorandom number generator. However, owing to Definition~\ref{def:secure-prng}, the two distributions are not efficiently distinguishable. If the values produced by $\rv{U}^{(m)}$ measurably deviated from the ideal distribution, then this could be used to distinguish between pseudorandom and truly random numbers, thus contradicting Definition~\ref{def:secure-prng}.

\section{Quantum estimation of stochastic moments}
\label{sec:estimation}

Quantum computers are naturally suited for estimating stochastic quantities, as measurements of quantum states inherently produce random outcomes. In this section, we demonstrate how to encode a linear quantity of interest derived from the transformed random field into the amplitude of a quantum state. This enables the use of quantum algorithms to estimate expected values.

We begin with a standard Monte Carlo approach, which has a runtime scaling of $\calO(\tol^{-2})$ for a target accuracy~$\tol$. We then show how this can be improved to $\calO(\tol^{-1})$ using amplitude estimation techniques. Finally, we discuss how to extend this framework to compute higher-order moments using nonlinear quantum algorithms developed in~\cite{DP24a}.

\subsection{Monte-Carlo sampling and amplitude estimation} \label{ssec:amplitude-estimation}

Let us consider a linear quantity of interest~
\begin{equation} \label{eq:qoi}
\lambda(\tgrf) = q_0 \tgrf(x_0) + \dots + q_{n-1} \tgrf(x_{n-1}),
\end{equation}
of the transformed random field~$\tgrf \colon \Omega \times D \to [-1, 1]$, with coefficients~$q \in \R^n$, suitably rescaled such that $|q|_1 = 1$.

The easiest way to classically estimate the expected value $\E[\lambda(\tgrf)]$ is an empirical mean over independent samples of the quantity of interest. An explicit error bound can be obtained, for example, via Hoeffding's inequality. It guarantees that an error of at most~$\tol > 0$ is achieved with probability at least $1 - \delta$ if $\tol^{-2} \log(2/\delta)$ independent samples are used. Here, we employed the assumption $|q|_1 = 1$ to bound $|\lambda(\tgrf)| \le 1$. Combining this with the cost of evaluating the quantity of interest for a single sample, we obtain an overall runtime proportional to~$\tol^{-2} n$ scaling linearly with the number of points~$n$.

On a quantum computer, we can compute all of these samples in parallel. Specifically, let $U_q$ be a quantum algorithm that prepares the state corresponding to the element-wise square root of the vector $q$, and let $U_{q,\pm}$ be a quantum algorithm that applies a phase flip for the indices where $q_j<0$, concretely
\begin{equation} \label{eq:prepare}
U_q\ket{0}_{\lceil \log n \rceil} = \sqrt{|q_0|} \ket{0} + \dots + \sqrt{|q_{n-1}|} \ket{n-1},
\qquad
U_{q,\pm}\ket{j} = \operatorname{sign}(q_j)\ket{j}.
\end{equation}
Then one can easily verify that
\[
\bra{0}(U_q \otimes \Id_a)^\dag \rv U^{(0)} (U_{q,\pm} U_q \otimes \Id_a) \ket{0} = \lambda(\rv{Y}),
\]
where $\rv U^{(0)}$ is the algorithm from Theorem~\ref{thm:main2}. To compute multiple samples at once, for some $m \in \N$ we apply the algorithm $\rv U^{(m)}$ to the superposition
\[
\ket{+}^{\otimes m} = \mathrm{H}^{\otimes m}\ket{0}_m = 2^{-m/2} (\ket{0} + \dots + \ket{2^m-1})
\]
of all sample indices, prepared using a number of Hadamard gates. This gives
\begin{equation} \label{eq:measurement-product}
\bra{0}(U_q \otimes \mathrm{H}^{\otimes m} \otimes \Id_a)^\dag \rv U^{(m)} (U_{q,\pm} U_q \otimes \mathrm{H}^{\otimes m} \otimes \Id_a) \ket{0}_{\lceil \log n \rceil + m + a} = 2^{-m} \sum_{k=0}^{2^m-1}\lambda(\tgrf^{(k)}).
\end{equation}
The right-hand side is exactly the mean of $2^m$ samples of $\lambda(\tgrf)$. By taking the number of samples as
\begin{equation} \label{eq:samples}
m \ge \log_2 \max\{1, \lceil \tol^{-2} \log(2/\delta)\rceil\},
\end{equation}
Hoeffding's inequality again guarantees an error of less than $\tol$ with probability at least $1 - \delta$.
The inner product~\eqref{eq:measurement-product}, however, still needs to be measured. The most efficient method to do so is \emph{amplitude estimation}.
\begin{lemma}[Quantum Amplitude Estimation \cite{BHMT02}] \label{lem:amplitude-estimation}
    Let $N \in \N$, $0 < \delta, \tol < 1$ and assume we have access to a quantum algorithm described by a unitary transformation~$U \in \C^{2^N \times 2^N}$, and let $\lambda \coloneqq \bra{0}U\ket{0}$.
    Then a classical estimate $\rv{\tilde \lambda}$ of $|\lambda|$ can be obtained with failure probability at most $\delta$, specifically
    \[
    \Pr(|\rv{\tilde \lambda} - |\lambda|| > \tol) < \delta,
    \]
    by using $U$ at most $\mathcal{O}(\tol^{-1}\log\delta^{-1})$ times.
\end{lemma}

Combining this lemma with~\eqref{eq:measurement-product} proves Theorem~\ref{thm:main3} for $s = 1$. The remaining step is to show how to compute mixed (and higher-order) moments.

\subsection{Mixed and higher order moments}
The previous section only deals with linear quantities of interest of the transformed random field~$\tgrf$. However, this restriction can be lifted. In fact, using ideas from nonlinear quantum computation developed in~\cite{DP24a}, we can efficiently estimate nonlinear mixed moments of degree~$s \in \N$ of the form
\[
\E[\lambda_1(\tgrf)\lambda_2(\tgrf) \cdots \lambda_s(\tgrf)],
\]
where each $\lambda_\ell$ for $\ell = 1, \dots, s$ is a linear functional as in~\eqref{eq:qoi}, but may have its own weights~$q_j^{(\ell)}$. 
Notably, this general class of observables includes important examples such as the covariance of the random field
\[\Cov(\tgrf(x_1), \tgrf(x_2))=\E[\tgrf(x_1)\tgrf(x_2)]-\E[\tgrf(x_1)]\E[\tgrf(x_2)]\]
at two spatial locations $x_1, x_2 \in D$. Note again, that this does not necessarily correspond to the original covariance~$c$ of the untransformed field~$\rv{Y}$.

Consider the quantum algorithm from~\eqref{eq:measurement-product}, but fix the index~$k \in [2^m]$ of the sample~$\tgrf^{(k)}$:
\[
(U_{q^{(\ell)}} \otimes \Id_{m + a})^\dag \rv U^{(m)}_\ell (U_{q^{(\ell)},\pm} U_{q^{(\ell)}} \ket{0}_{\lceil\log n\rceil} \otimes \ket{k}_m\ket{0}_a) = \left(\lambda_\ell(\tgrf^{(k)})\ket{0} + \ket{*}\right) \otimes \ket{k},
\]
where $\ket{\ast}$ is a short-hand notation to represent states orthogonal to $\ket{0}$, capturing all residual terms irrelevant to the measurement outcome.
Now apply this procedure for each $\ell = 1, \dots, s$ in parallel, using separate quantum registers for the output of each $\lambda_\ell$ while reusing the shared register holding the sample index~$\ket{k}$. The resulting joint circuit, denoted by $U^{(m)}_\mathrm{total}$ (see Figure~\ref{fig:circuit-nonlinear}), prepares the state
\begin{align*}
U^{(m)}_\textrm{total}\ket{0}_{s(\lceil \log n \rceil + a)}\ket{k}_m &= \left(\lambda_1(\tgrf^{(k)})\ket{0}_{\lceil \log n \rceil + a} + \ket{\ast}\right) \otimes \dots \otimes\left(\lambda_s(\tgrf^{(k)})\ket{0} + \ket{\ast}\right) \otimes \ket{k} \\
&= \left(\lambda_1(\tgrf^{(k)})\cdots\lambda_s(\tgrf^{(k)})\ket{0}_{s(\lceil \log n \rceil + a)} + \ket{\ast}\right)\ket{k}.
\end{align*}
Thus, the amplitude of the state $\ket{0}_{s(\lceil \log n \rceil + a)}\ket{k}_m$ is precisely the value of the mixed moment observable for the sample $\tgrf^{(k)}$.
To average over $2^m$ independent samples, we again prepare the uniform superposition $\ket{+}^{\otimes m}$ and compute
\[
\bra{0}\bra{+}^{\otimes m}U^{(m)}_\textrm{total} \ket{0}\ket{+}^{\otimes m} = 2^{-m} \sum_{k=0}^{2^m-1}\lambda_1(\tgrf^{(k)})\cdots\lambda_s(\tgrf^{(k)}).
\]
Since the mixed quantity of interest is bounded by $[-1,1]$, the error estimate~\eqref{eq:samples} still holds. Finally, each $\rv U^{(m)}_\ell$ for $\ell = 1, \dots, s$ is called once, which, together with amplitude estimation from Lemma~\ref{lem:amplitude-estimation}, leads to the runtime bound in Theorem~\ref{thm:main3}.
\begin{figure}
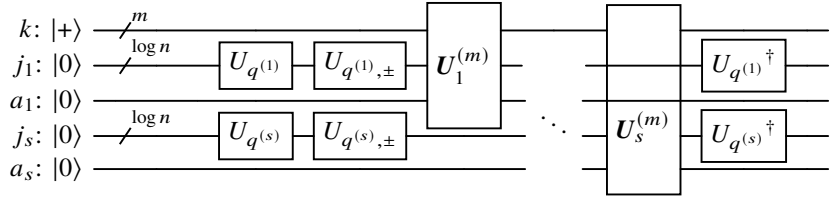

    \centering
    \includestandalone{circuit_nonlinear}
    \caption{The circuit $U^{(m)}_\textrm{total}$ to compute mixed moments. }
    \label{fig:circuit-nonlinear}
\end{figure}

A simpler approach for computing powers  of a single quantity of interest, i.e., higher-order moments, can also be formulated using Quantum Singular Value Transformation (QSVT)~\cite{GSLW19}, specifically using constructions for nonlinear quantum computation introduced in \cite{GMF21,RR23}

\begin{remark}[Application to PDEs with random coefficients] \label{rem:pdes}
A potentially important application of the proposed quantum algorithm is the numerical solution of linear second-order divergence form PDEs with random coefficients, such as the prototypical stationary diffusion equation
\begin{equation} \label{eq:stochastic-pde}
-\operatorname{div}(\tgrf \nabla \rv{u}) = f \qquad \text{in }D
\end{equation}
with homogeneous Dirichlet boundary condition $\rv{u} = 0$ on $\partial D$. To ensure ellipticity of the scalar random coefficient~$\tgrf$, the rescaling function~$\rho$ must map into a positive interval $[\alpha,1]$ for some uniform lower bound~$\alpha > 0$. 

The deterministic version of this problem has been analyzed in~\cite{DP25}, where a block encoding of the diagonal matrix $\Diag(\tgrf(x_j))$ was used as input to a quantum linear system algorithm. Theorem~\ref{thm:main2} provides exactly this block encoding in the stochastic setting.

This, in principle, enables the application of similar techniques to estimate moments of quantities of interest such as $(q, \rv{u})_{L^2}$ for a given $q \in L^2(D)$. A full treatment of this application, including quantum implementation and error analysis, lies beyond the scope of the present paper and will be addressed in future work.
\end{remark}

\section{Numerical experiments}
To illustrate our theoretical results, we present numerical experiments based on both classical and quantum implementations. These include a classical evaluation of the algorithm underlying Theorem~\ref{thm:main1}, as well as simulations of the quantum algorithms from Theorems~\ref{thm:main2} and~\ref{thm:main3} on emulated quantum hardware. All code is implemented using the \texttt{qiskit} framework~\cite{JTK+24} and is available at \href{https://github.com/MDeiml/quantum-grf}{\texttt{github.com/MDeiml/quantum-grf}}. The repository also contains some additional figures.

\subsection{Convergence of the random field}
We illustrate the convergence behavior established in Lemma~\ref{lem:quadrature} by comparing our approximation against a high-fidelity reference solution. The domain of interest is $D = [0, 1]^2$, and we choose a Gaussian covariance function of the form~\eqref{eq:gaussian-covariance} with parameters $\xi = 0.075$ and $C = 1$.

To obtain a reference implementation for comparison, we compute high-resolution samples using the discrete representation~\eqref{eq:discrete-random-field}, where for convenience of implementation we choose a truncation ball with respect to the maximum norm.
Specifically, we take a large oversampling domain~$\tilde D \coloneqq [-5, 6]^2 \supset D$, and discretize it using a fine grid of step size $h = 1/240$, in the sense that we generate independent standard normal samples~$\wn_j$ for each node of the grid $h\Z^2 \cap \tilde D$. The reference values of the random field are then computed at points $x \in h\Z^2 \cap D$. Due to this conforming choice of quadrature and evaluation points, the computation for this reference is equivalent to more typical moving averages or FFT based sampling.
This discrete model serves as a surrogate for the exact solution, allowing a meaningful assessment of the accuracy of our interpolation-based method. A sample of the resulting high-resolution Gaussian random field is shown at the bottom left of Figure~\ref{fig:realizations}. 

To perform the convergence analysis, we require noise variables on a coarser grid, say by a factor~$\ell \in \N$. For this, we define
\[
\wn_j^{(\ell)} = \; \ell^{-1} \; \sum_{\mathclap{k\in\Z^2\cap[-5/h,6/h]^2}}\quad\sinc\left(\frac{k}{\ell} - j\right)\wn_k,
\]
which resembles a trapezoidal quadrature of~\eqref{eq:white-noise-sinc} for a grid size of $\ell h$.
Indeed, one can check that, if the fine noise~$\wn_j$ were to be computed from the white noise integral~\eqref{eq:white-noise-sinc} and the oversampling domain was infinitely big, then this quadrature would be exact.
This lets us compute $\rv Y_{r\!,\ell h}$ as
\[
\rv Y_{r\!,\ell h}(x) = \quad \sum_{\mathclap{j\in\Z^2\cap[-r,r]^2}} \quad f(x, j\ell h) \wn_j^{(\ell)},
\]
for which realizations can be directly compared to $\rv{Y}$.
Again, we choose a truncation ball with respect to the maximum norm, this time with the radius $r = 2\xi^2\pi/(\ell h)^2$. This should yield an asymptotic error of $\calO(e^{-\pi r})$ specifically for the Gaussian covariance; cf.~\cite[Table~7.1]{TW14} and Remark~\ref{rem:correlation-length}.
We evaluate this field at the fine grid points of the discrete reference model. A sample for $\ell = 24$ is depicted at the bottom right of Figure~\ref{fig:realizations}. It is visually very similar to the reference. To illustrate the result of Lemma~\ref{lem:quadrature}, we compute $100$ samples for the reference field and approximations with $\ell$ ranging from $4$ to $24$. The resulting convergence plot in Figure~\ref{fig:convergence} shows a linear relationship between the logarithm of the expected maximum error and the truncation radius~$r$ (up to the point where machine precision is likely to affect the error). This behavior is in agreement with the theoretical prediction of Theorem~\ref{thm:main1}.

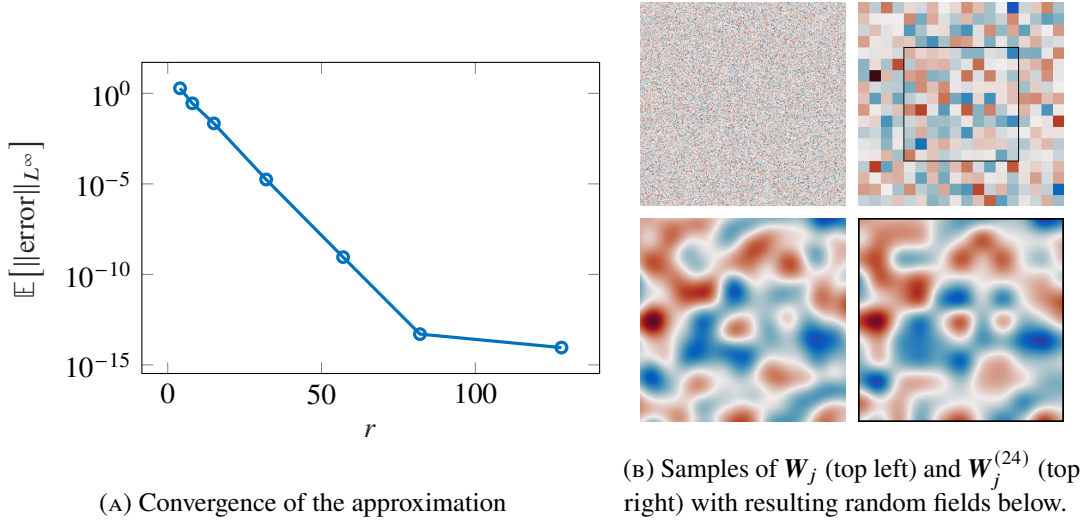
\begin{figure}
    \centering
    \begin{subfigure}{0.55\linewidth}
    \centering
    \begin{tikzpicture}
        \begin{axis}[%
		width=0.73\textwidth,
		height=0.5\textwidth,
		at={(0\textwidth,0\textwidth)},
		scale only axis,
		unbounded coords=jump,
		xlabel style={font=\color{white!15!black}},
		xlabel={$r$},
		ylabel style={font=\color{white!15!black}},
		ylabel={$\E\big[\|\text{error}\|_{L^\infty}\big]$},
        ymode = log,
		axis background/.style={fill=white},
		]
		\addplot [mark=o, color=mycolor1, very thick, mark size=2pt] table [col sep=semicolon, x index = 0, y index = 1] {data/moving_averages_error.csv};
        \end{axis}
    \end{tikzpicture}
    \vspace{2ex}
    \subcaption{Convergence of the approximation}
    \label{fig:convergence}
    \end{subfigure}
    \begin{subfigure}{0.4\linewidth}
    \centering
    \setlength{\tabcolsep}{0.5ex}
    \begin{tabular}{cc}
    \includegraphics[width=0.45\linewidth]{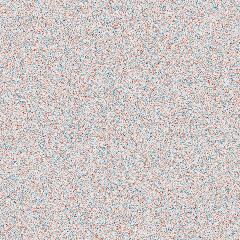} & \includegraphics[width=0.45\linewidth]{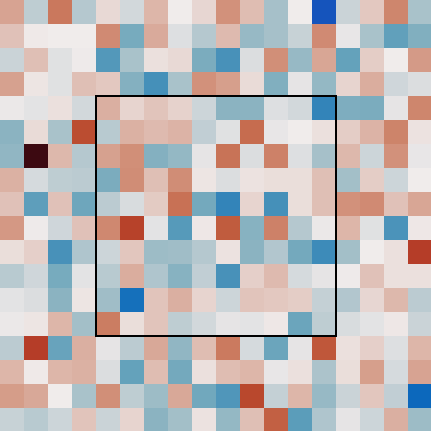} \\
        \includegraphics[width=0.45\linewidth]{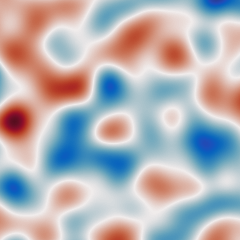} & \includegraphics[width=0.45\linewidth]{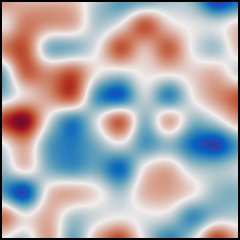}
    \end{tabular}
    \subcaption{Samples of $\wn_j$ (top left) and $\wn_j^{(24)}$ (top right) with resulting random fields below. }
    \label{fig:realizations}
    \end{subfigure}
    \caption{
    Accuracy of the approximation $\rv Y_{r\!,\ell h}$ from Lemma~\ref{lem:quadrature}, for a random field with Gaussian covariance~\eqref{eq:gaussian-covariance} and $\xi = 0.075$. The grid size is chosen as $h = \xi\sqrt{2\pi/r}$. Left: Convergence with respect to the truncation radius~$r$. Right: A realization of the reference random field~$\rv{Y}$ and the white noise it is computed from, along with the corresponding approximation with $\ell = 24$. The coarse noise is computed from the fine noise and includes $r = 4$ additional pixels around the original domain (shown as a black frame). The noise used for the reference is much finer and may therefore appear almost uniform gray.
    }
\end{figure}

\subsection{Quantum uncertainty estimation of random fields}

We implement a PCG random generator similar to the one described in~\cite{ON14}. For simplicity, we restrict ourselves to the state size~$\Nseed = 6$. Each generated number consists of $4$ bits, resulting in a total of $4 \cdot 2^6 = 2^8$ distinct outputs. All output bits of this generator for a sample seed are shown in Figure~\ref{fig:pcg}. Due to the small size of the state space, this generator cannot be expected to pass any standard statistical randomness tests. Nevertheless, we verified that the quantum implementation reproduces the exact output of its classical counterpart, confirming the correctness of the circuit.

\begin{figure}
    \centering
    \begin{subfigure}{0.28\linewidth}
    \centering
    \includegraphics[width=0.9\linewidth]{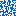}
    \caption{Random bits from PCG.}
    \label{fig:pcg}
    \end{subfigure}
    \hspace{0.04\linewidth}
    \begin{subfigure}{0.28\linewidth}
    \centering
    \includegraphics[width=0.9\linewidth]{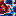}
    \caption{$\tgrf$ from quantum circuit.}
    \label{fig:qc_sample}
    \end{subfigure}
    \hspace{0.04\linewidth}
    \begin{subfigure}{0.28\linewidth}
    \centering
    \includegraphics[width=0.9\linewidth]{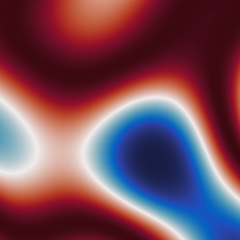}
    \caption{$\tgrf$ reference.}
    \label{fig:sin}
    \end{subfigure}
    \caption{Samples of different random fields. 
Left~(\textsc{a}): uncorrelated random bits from the quantum PCG. 
Right~(\textsc{c}): high-resolution sample of $\tgrf=\cos\rv{Y}$ from classical implementation.
Middle~(\textsc{b}): grid values of another realization of $\tgrf$ from quantum circuit, showing visually similar structures.}
\end{figure}

We then use the random bits generated by this PCG to construct the diagonal-scaling quantum circuit described in Theorem~\ref{thm:main2}. For simplicity, the normally distributed random variables~$\wn_j$ are not generated via the inverse transform~\eqref{eq:pdf}, but instead approximated by the empirical average
\[
\wn_j = n^{-1/2}\sum_{k = 1}^n \big(1 - 2 \rv{b}_k^{(j)}\big),
\]
where each $\rv{b}_k^{(j)}$ is a random bit. This construction converges to a standard normal distribution as $n \to \infty$ by the Central Limit Theorem; in our implementation, we use $n = 4$. Importantly, this approximation does not affect the covariance structure of the resulting Gaussian random field. In fact, as shown in~\cite[Appendix~C]{GBHH94}, this procedure becomes exact for any $n$ in the limit $h \to 0$.

As covariance, we choose the Gaussian covariance~\eqref{eq:gaussian-covariance} with parameters $\xi^2 = 1/8$ and $C = 1$, on the unit square domain $D = [0, 1]^2$. For simplicity, we select the transformation $\rho = \cos$. Thanks to $\cos^{-1}$ being used in the construction of Lemma~\ref{lem:amplitude-encoding}, this choice renders the computation linear in the values~$\wn_j$.
Figure~\ref{fig:sin} shows a sample of the transformed random field~$\tgrf$ generated classically, while Figure~\ref{fig:qc_sample} illustrates a sample evaluated at the grid points $[16]^2/16$ using the quantum circuit.

Using the construction of Theorem~\ref{thm:main3} we now approximate $\Cov(\tgrf_\mathrm{left}, \tgrf_\mathrm{right})$, where
\[
\tgrf_\mathrm{left} = 2 \int\limits_0^{1/2} \! \int\limits_0^1 \tgrf(x, y) \,\mathrm{d}y \dx \approx \frac{1}{128}\sum_{x = 0}^{7} \sum_{y = 0}^{15} \tgrf(x/16, y/16)
\]
is the average of $\tgrf$ over the left half of the domain, and similarly, $\tgrf_\mathrm{right}$ is the average over the right half. A classical numerical approximation gives a reference value of $0.373$ for this covariance. Running the corresponding quantum circuit in a simulator with $2^m = 32$ samples yields a qubit amplitude of approximately $0.29$, corresponding to an error of about $0.08$. This amplitude can be estimated to within an error of about $0.01$ using $10{,}000$ shots of the algorithm via standard Monte Carlo estimation. This is in line with our theoretical predictions.

\section{Conclusion}
We proposed and analyzed an improved variant of the classical moving average scheme for generating Gaussian random fields. In the quantum setting, this leads to an efficient algorithm for producing, for example, random coefficients for partial differential equations. A key advantage of the method is that it allows pointwise nonlinear transformations of the random field. Through numerical experiments, we demonstrated the validity and efficiency of the approach. 

The method has one key limitation: it relies on access to a highly regular (in fact, analytic) convolution kernel. In many practical scenarios, however, only the covariance function is specified, and the efficient computation of the corresponding convolution kernel remains unclear. For the special but significant case of Gaussian covariance, the kernel is explicitly known and sufficiently smooth, making the method directly applicable. Extending the approach to more general covariance structures remains a central challenge for future work.

\section*{Acknowledgements}
We thank Martin Hermann and Michael Feischl for fruitful discussion on the classical generation and sampling of Gaussian random fields.

%% file: circuit_rotation.tex
    \begin{quantikz}[row sep={1.2em,between origins}, column sep=0.6em]
    \lstick{$\ket{j}$} & \qwbundle{\log n} & \phantomgate{\hspace{3ex}} & \gate[3]{\oplus (\cos^{-1}(\mathcal{C}(j)))} & & \midstick[1, brackets=none]{$\cdots$} & & \gate[3]{\oplus (\cos^{-1}(\mathcal{C}(j)))} & \\
    \lstick{$\ket{k_1}$} & & & & \ctrl{2} & \midstick[2, brackets=none]{$\ddots$} & & &\\
    \lstick{$\ket{k_m}$} & & & & & & \ctrl{1} & &\\
    & & & & \gate{\mathrm{R}_\mathrm{Y}(2^{-m}\pi)} & \midstick[1, brackets=none]{$\cdots$} & \gate{\mathrm{R}_\mathrm{Y}(2^{-1}\pi)} & &
    \end{quantikz}

%% file: circuit_nonlinear.tex
    \scalebox{0.9}{
    \begin{quantikz}[row sep=-0.7em, column sep=0.8em, transparent]
    \lstick{$k$: $\ket{+}$} & \qwbundle{m} & & & \gate[3]{\bm{U}_1^{(m)}} & & \gate[5,label style={yshift=-0.4cm}]{\bm{U}_s^{(m)}} & & & \\
    \lstick{$j_1$: $\ket{0}$} & \qwbundle{\log n} & \gate{U_{q^{(1)}}} & \gate{U_{q^{(1)},\pm}} & & \midstick[4, brackets=none]{$\ddots$} & \linethrough & \gate{U_{q^{(1)}}{}^\dag} & & \\
    \lstick{$a_1$: $\ket{0}$} & \phantomgate{\hspace{7ex}} & & & & & \linethrough & & & \\
    \lstick{$j_s$: $\ket{0}$} & \qwbundle{\log n} & \gate{U_{q^{(s)}}} & \gate{U_{q^{(s)},\pm}} & & & & \gate{U_{q^{(s)}}{}^\dag} & & \\
    \lstick{$a_s$: $\ket{0}$}& & & & & & & & &
    \end{quantikz}
    }

%% file: appendix.tex
\section{Trapezoidal rule in multiple dimensions}
\label{sec:multiple-dimensions-proof}

We roughly outline how to extend the proof for exponential convergence of the trapezoidal rule, as can be found in \cite{TW14}, to multiple dimensions.
For this, let $g \colon \R^d \to \R$ be a function such that $g$ extends to an analytic function $\tilde g \colon (\R + (-\alpha_g, \alpha_g)i)^d \to \C$ and there exist constants $\beta_g, \gamma_g > 0$ with
\[
|g(x)| \le \gamma_g e^{-\beta_g \lvert\Re x|} \qquad\forall x \in (\R + (-\alpha_g, \alpha_g)i)^d.
\]
Let us first check the one-dimensional integral
\begin{equation} \label{eq:helper-exponential}
\int_{-\infty}^\infty \gamma_g e^{-\beta_g|x|} \dx = 2\gamma_g\int_{0}^\infty e^{-\beta_gx} \dx = 2\frac{\gamma_g}{\beta_g}.
\end{equation}
We consider the integral
\(
I = \int_{\R^d} g(x) \dx,
\)
which exists due to the upper bound. We show iteratively that this is approximated by
\[
I_n \coloneqq h^n \sum_{j_1 \in \Z} \dots \sum_{j_n \in \Z} \iint_{-\infty}^\infty g(hj_1, hj_2, \dots, hj_n, x_{n+1}, \dots, x_d) \,\mathrm{d}x_{n+1}\dots x_d.
\]
Clearly, we have $I = I_0$, and $I_d$ corresponds to the $d$-dimensional trapezoidal rule. Due to the one-dimensional trapezoidal rule, we have
\[
|I_n - I_{n+1}| \le h^n \sum_{j_1 \in \Z} \dots \sum_{j_n \in \Z} \iint_{-\infty}^\infty \frac{2 M(hj_1, hj_2, \dots, hj_n, x_{n+1}, \dots, x_d) }{e^{2\pi \alpha_g/h} - 1} \,\mathrm{d}x_{n+2}\dots x_d
\]
where we applied the rule to the function
\(
y \mapsto g(hj_1, \dots, hj_n, y, x_{n+2}, \dots, x_d).
\)
Here,
\[
M(x) = \int_{-\infty}^\infty \gamma_g e^{-\beta_g\sqrt{|x|^2 + y^2}} \,\mathrm{d}y \le \frac{2\gamma_gd}{\beta_g}e^{-\beta_g |x_1|/d}\dots e^{-\beta_g|x_{d-1}|/d}.
\]
For the inequality we used~\eqref{eq:helper-exponential}. We may use~\eqref{eq:helper-exponential} another $d-1$ times to show
\[
|I - I_d| \le \sum_{n \in [d]} |I_n - I_{n+1}| \le d\frac{2(2\gamma_gd)^d}{\beta_g^d(e^{2\pi \alpha_g/h} - 1)}.
\]

\section{Bound for the analytic continuation of the sinc function} \label{sec:sinc-bound}

For the proof of Lemma~\ref{lem:quadrature} we require a bound on the norm of the analytic continuation of the sinc function. The norm of this analytic continuation is given by
\[
|\sinc(x + iy)| = \frac{|\sin((x + iy)\pi)}{\pi|x + iy|} = \frac{|e^{(ix-y)\pi} - e^{-(ix - y)\pi}|}{2\pi|x + iy|} = \frac{|e^{-y\pi}e^{i\pi x} - e^{y\pi}e^{-i\pi x}|}{2\pi|x + iy|}.
\]
With this we calculate
\begin{align*}
|\sinc(x + yi)| &\le \frac{e^{-y\pi}|e^{i\pi x} - e^{-i\pi x}| + |e^{y\pi} - e^{-y\pi}|}{2\pi|x + yi|} \le e^{-\pi y}\frac{\sin(\pi x)}{\pi|x|} + \frac{|e^{-\pi y} - 1 + 1 - e^{\pi y}|}{2 \pi |y|} \\
&\le e^{\pi|y|}|\sinc(x)| + \frac{2|e^{|y|\pi} - 1|}{2\pi|y|} \le 2e^{\pi|y|}.
\end{align*}